\pdfoutput=1
\documentclass[envcountsame]{llncs}
\spnewtheorem{definition}{Definition}{\bfseries}{\rmfamily}

\usepackage{amssymb,latexsym}

\usepackage{amsthm}
\usepackage{adjustbox}
\usepackage{algorithm2e}
\usepackage{enumerate}
\usepackage{appendix}
\usepackage{amsmath}
\usepackage[numbers]{natbib}
\makeatletter

\newtheorem{myprob}[definition]{Problem}
\renewcommand\bibsection%
{
  \section*{\refname
  \@mkboth{\MakeUppercase{\refname}}{\MakeUppercase{\refname}}}
}
\makeatother

\usepackage{mathtools}
\usepackage{mathrsfs}

\usepackage{etoolbox}
\makeatletter
    \pretocmd{\NAT@citexnum}{\@ifnum{\NAT@ctype>\z@}{\let\NAT@hyper@\relax}{}}{}{}
\makeatother

\usepackage{xspace} 

\usepackage{pgf,tikz}
\usepackage{wrapfig}  
\usepackage{cutwin}  

\usepackage{hyperref}
\hypersetup{%
   breaklinks,%
   colorlinks=true,%
   linkcolor=[rgb]{0.45,0.0,0.0},%
   urlcolor=[rgb]{0.05,0.390,0.0},%
   citecolor=[rgb]{0,0,0.45}
}%

\usepackage[noend]{algorithmic}

\usetikzlibrary{arrows,shapes,calc,intersections,through,backgrounds,patterns}

\theoremstyle{plain}
\newtheorem{obs}[theorem]{Observation}
\newtheorem{Claim}[theorem]{Claim}

\usepackage{thmtools, thm-restate} 

\def\T{{\cal T}}

\mainmatter
\begin{document}
\title{Polynomial Time Algorithms for Bichromatic Problems}


\author{Sayan Bandyapadhyay\inst{1}\thanks{sayan-bandyapadhyay@uiowa.edu}
        \and Aritra Banik\inst{2}\thanks{aritrabanik@gmail.com
}
        }
\institute{
  Computer Science, University of Iowa, Iowa City, USA
  \and  Department of Computer Science \& Engineering, Indian Institute of Technology, Jodhpur, India
  }
  
\authorrunning{S.\,Bandyapadhyay, A.\,Banik} 

\maketitle

\vspace{-0.4cm}
\begin{abstract}
In this article, we consider a collection of geometric problems involving points colored by two colors (red and blue), referred to as bichromatic problems. The motivation behind studying these problems is two fold; (i) these problems appear naturally and frequently in the fields like Machine learning, Data mining, and so on, and (ii) we are interested in extending the algorithms and techniques for single point set (monochromatic) problems to bichromatic case. For all the problems considered in this paper, we design low polynomial time exact algorithms. These algorithms are based on novel techniques which might be of independent interest.
\end{abstract} 

\section{Introduction}\label{sec:intro}

In discrete and computational geometry one of the most important classes of problems are those involving points colored by two colors (red and blue), referred to as bichromatic problems. These problems have vast applications in the fields of Machine learning, Data mining, Computer graphics, and so on. For example, one natural problem in learning and clustering is separability problem \cite{Cristianini,Duda}, where given a red and a blue set of points, the goal is to separate as many red (desirable) points as possible from the blue (non-desirable) points using geometric objects. Another motivation to study bichromatic problems is to extend the algorithms and techniques in discrete and computational geometry for single point set problems to bichromatic case (see the survey by Kaneko and Kano \cite{kaneko} and some more recent works \cite{arkincccg,arkin2015choice,arkin15,BiniazBMS16,BiniazMNS15}).

In this paper we consider two bichromatic problems that arise naturally in practice. Throughout the paper as coloring we refer to a function that maps points to the range $\{red,blue\}$. In the first problem, we are given two finite disjoint sets of points $R$ and $B$ in the plane, colored by red and blue, respectively. Denote the respective cardinality of $R$ and $B$ by $n$ and $m$. Also let $\T=R\cup B$. In the second problem, we are given a finite collection $Q$ of pairs of points in the plane and we need to find a coloring of the points using red and blue such that certain optimality criterion is satisfied. For the sake of simplicity of exposition we assume that all the points are in general position. Assuming this, we formally define the two problems mentioned before. 
\paragraph{Maximum Red Rectangle Problem.} A rectangle (of arbitrary orientation) is called \textit{red} if it does not contain any blue points in its interior\footnote{the red rectangle may contain blue points on its boundary}. The \textit{size} of a red rectangle is defined as the number of red points contained in it. The \textit{maximum red rectangle problem} (MRR) is to find a red rectangle of maximum size. Such a rectangle will be referred to as a \textit{maximum red rectangle}. We note that the version where the rectangles are constrained to be axes parallel is a special case of MRR. 

Liu and Nediak \cite{LiuN03} considered the axes parallel version of MRR and designed an algorithm that runs in $O(n^2\log n+nm+m\log m)$ time and $O(n)$ space. Backer and Keil \cite{BackerK09} improved this time bound to $O((n+m){\log}^3 (n+m))$ using a divide-and-conquer approach due to Aggarwal and Suri \cite{AggarwalS87}. However, their algorithm runs in $O(n\log n)$ space. For axes parallel squares, they designed an $O((n+m){\log} (n+m))$ time algorithm.
As far as we are concerned the general version of MRR was not considered before. However, two related problems involving rectangles of arbitrary orientation have been studied before. The first one is the largest empty rectangle problem (LER) where given a point set $P$, the goal is to find a rectangle of the maximum area which does not contain any point of $P$ in its interior \cite{AggarwalS87,ChaudhuriND03,ChazelleDL86,Naamad84,Orlowski90}. This problem can be solved in $O(|P|^3)$ time with $O(|P|^2)$ space \cite{ChaudhuriND03}. The second problem is a bichromatic problem, where the goal is to find the rectangle that contains all the red points, the minimum number of blue points and has the largest area. This problem can be solved in $O(m^3+n\log n)$ time \cite{ArmaseluD16}.

Several variants of MRR have been studied in the past. Aronov and Har-Peled \cite{AronovH08} considered the problem of finding a maximum red disk and gave a $(1-\epsilon)$-factor approximation algorithm with $O((n+m){\epsilon}^{-2}{\log}^2 (n+m))$ expected running time. Eckstein~\emph{et al.} \cite{EcksteinHLNS02} considered a variant of MRR for axes parallel hyperrectangles in high dimensions. They showed that, if the dimension of the space is not fixed, the problem is NP-hard. However, they presented an $O(n^{2d+1})$ time algorithm, for any fixed dimension $d \geq 3$. Later, Backer and Keil \cite{BackerK10} have significantly improved this time bound to $O(n^d{\log}^{d-2} n)$.  
Bitner~\emph{et al.} \cite{BitnerCD10} have studied the problem of computing all circles that contain the red points and as few blue points as possible in its interior. 
In \cite{CortesDPSUV09} Cortes~\emph{et al.} have considered the following problem: find a largest subset of $R\cup B$ which can be enclosed by the union of two not necessarily disjoint, axes-aligned rectangles $R_1$ and $R_2$ such that $R_1$ (resp. $R_2$) contains only red (resp. blue) points. 

\paragraph{Maximum Coloring Problem.} We are given a collection $Q$=$\{(a_1,b_1),\ldots, (a_n,b_n)\}$ of pairs of points. Let $P=\cup_{i=1}^n \{a_i,b_i\}$. A coloring of the points in $P$ is \textit{valid} if for each pair of points in $Q$, exactly one point is colored by blue and the other point is colored by red. Now consider any valid coloring $C$. For any halfplane $h$, we denote the set of red points and blue points in $h$ by $h_r(C)$ and $h_b(C)$, respectively. Let $\mathcal{H}({C})$ be the set of all halfplanes $h$ such that $|h_b(C)|=0$. Define $\eta({C})=\max_{h\in \mathcal{H}({C})} |h_r(C)|$. The \textit{Maximum Coloring Problem} (MaxCol) is to find a valid coloring ${C}$ that maximizes $\eta({C})$.

Problems related to coloring of points have been studied in literature. Even~\emph{et al.} \cite{EvenLRS03} defined Conflict-free coloring where given a set $X$ of points and a set of geometric regions, the goal is to color the points with minimum number of colors so that for any region $r$, at least one of the points of $X$ that lie in $r$ has a unique color. This problem have further been studied in offline \cite{CheilarisGRS14,katz2012conflict} and online \cite{Bar-NoyCS08,ChenFKLMMPSSWW07} settings.

\paragraph{Our Results.} We give exact algorithms for both of the problems considered in this paper. In Section \ref{sec:mrr}, we design an algorithm for MRR that runs in $O(g(n,m)\log (m+n)+n^2)$ time with $O(n^2+m^2)$ space, where $g(n,m)\in O(m^2(n+m))$ and $g(n,m)\in \Omega(m^2+mn)$.
To solve the problem in general case, we need to solve an interesting subproblem which can be of independent interest. The nontriviality of this problem arises mainly due to the arbitrary orientations of the rectangles.


In Section \ref{sec:mcp}, we show that MaxCol can be solved in $O (n^{\frac{4}{3}+\epsilon}\log n)$ time. In particular, we show a linear time reduction from MaxCol to a problem and design an algorithm for the latter problem with the desired time complexity.

\section{Maximum Red Rectangle Problem}\label{sec:mrr}

Before stepping into the general case let us consider the axes parallel version. We note that the ideas mentioned here for the axes parallel case have been noted before in \cite{BackerK09}.

\subsection{The Axes Parallel Case}
Note that the number of red rectangles can be infinite. But for the sake of computing a maximum red rectangle we can focus on the following set. Consider the set $S$ of red rectangles such that each side of any such rectangle contains a point of $B$ or is unbounded. We note that the candidate set for the axes parallel version of Largest Empty Rectangle problem (LER) on $B$ is exactly $S$. Using this connection we will use several results from the literature of LER. Namaad~\emph{et al.} \cite{ChazelleDL86} proved that $|S|=O(m^2)$ and in expected case $|S|=O(m\log m)$. Orlowski \cite{Orlowski90} has designed an algorithm that computes the set $S$ in $O(|S|)$ time. The algorithm requires two sorted orderings of $B$, one with respect to $x$-coordinates and the other with respect to $y$-coordinates. 

To compute a maximum red rectangle we use Orlowski's algorithm. In each iteration when the algorithm computes a rectangle, we make a query for the number of red points inside the rectangle. Lastly, we return the rectangle that contains the maximum number of red points. Now, using $O(n\log n)$ preprocessing time and space, one can create a data structure that can handle orthogonal rectangular query in $O(\log n)$ time. Hence the axes parallel version of MPP can be solved in $O(|S|\log n+n\log n+m\log m)$ time with $O(n\log n+m)$ space required.
%

\subsection{The General Case}

For a point $p$, we denote its $x$ and $y$ coordinate by $x_p$ and $y_p$, respectively. A rectangle is said to be \textit{anchored} by a point set $P$ if it contains the points of $P$ on its boundary. For any red rectangle $T$, denote $R\cap T$ by $R_T$ and its boundary by $\delta(T)$. 

Like in the axes parallel case, in general case also the number of red rectangles can be infinite. However, we will define a finite subset of those rectangles such that the subset contains a maximum red rectangle. We note that this is a key step based on which we design our algorithm. The approach is not that subtle like the one for axes parallel case. 

Consider the set of red rectangles $C$ with the following two properties. For any rectangle $T\in C$, (i) at least one side of $T$ contains two points $p,q$ such that $p\in B$ and $q\in R\cup B$; and (ii) each of the other sides of $T$ either contains a point of $B$ or is unbounded. Now we have the following lemma.

\begin{restatable}{lemma}{cardofc}\label{obs:cardofc}
 $|C|=O(m^2(n+m))$ and $|C|=\Omega(m^2+mn)$.
\end{restatable}

\begin{proof}
 $C$ contains rectangles of two categories; (i) one which are unbounded from three sides, (ii) the others which are bounded from at least two sides. As each rectangle of first category can be identified by a point of $B$ and a point of $R\cup B$, the number of such rectangles is $O(m(n+m))$. Now consider a rectangle $T$ of second category. At first suppose that $T$ is bounded from two opposite sides. Then there are three points $p,q$ and $r$ such that $p,r \in B$, $q\in R\cup B$, and $p,q$ are contained on the side of $T$ opposite to the one containing $r$. Note that given the information that $p,q$ and $r$ are contained on the two opposite sides of a rectangle, we can uniquely form the rectangle. Thus there are $O(m^2(n+m))$ rectangles of second category which are bounded from two opposite sides. Similarly, one can show that there are $O(m^2(n+m))$ rectangles of second category which are bounded from two consecutive sides. Thus $|C|=O(m^2(n+m))$. Also $|C|=\Omega(m^2+mn)$, as for any two points $p,q$ such that $p\in B$ and $q\in R\cup B$ there is at least one rectangle in $C$. 
\end{proof}

The next lemma shows that $C$ is indeed a good candidate set for finding a maximum red rectangle. 

\begin{restatable}{lemma}{candi}\label{candidate}
 The set $C$ as defined above contains a maximum red rectangle.
\end{restatable}

\begin{figure*}
  \centering
  \includegraphics[width=42mm]
    {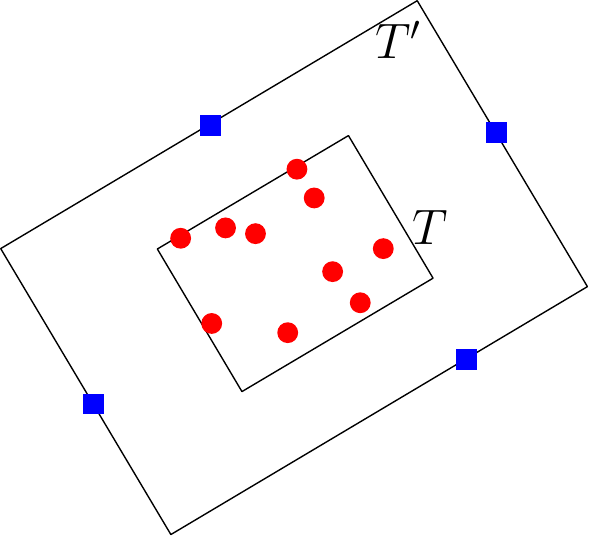}\\
  \caption{A maximum red rectangle $T$ and a rectangle $T'$ in $U$ containing it. Red (resp. Blue) points are shown by disks (resp. squares).}
  \label{fig:setU}
\end{figure*}

\begin{proof}
 Consider any maximum red rectangle $T$. Also consider the set of red rectangles $U$ such that for any rectangle $T_1 \in U$, no red rectangle can properly enclose $T_1$. It is easy to note that for any rectangle in $U$, each of its side either contains a point of $B$ or is unbounded. Indeed, any such rectangle is bounded from at least one side, or it contains the blue points inside it which is not possible. As $T$ is a red rectangle, by definition of $U$, $T$ must be contained inside a rectangle of $U$, say $T'$ (see Figure \ref{fig:setU}). Also $R_{T'}$ is equal to $R_T$. If not, then as $T'$ contains $T$, $R_T \subset R_{T'}$, which violates the assumption that $T$ is a maximum red rectangle. Thus $T'$ is a maximum red rectangle and if it is in $C$ we are done. Thus assume that $T' \notin C$.
 
 By definition of $C$, $T'$ contains exactly one blue point and no red point on each of its bounded sides. We rotate $T'$ anticlockwise so that any blue point on any of its bounded sides always remains on that side (see Figure \ref{fig:rectrotate}(a)). Note that if we keep on rotating $T'$ in this fashion, then two cases are possible; (i) a side of $T'$ touches a point of $R_{T'}$ or a point of $B$ which was not on $\delta(T')$ before rotation (see Figure \ref{fig:rectrotate}(b)), and (ii) a blue point on $\delta(T')$ becomes one of the corner points of $T'$ (see Figure \ref{fig:rectrotate}(c)). In both of the cases the rectangle $T'$ contains two points $p,q$ on one of its sides such that $p\in B$ and $q\in R\cup B$, and each of the other sides of $T$ either contains a point of $B$ or is unbounded. Thus $T'$ is in $C$. As $R_{T'}$ is equal to $R_T$, $T'$ is a maximum red rectangle which completes the proof. 
%
\end{proof}

\begin{figure*}[h]
 \centering
  \begin{minipage}[c]{0.33\textwidth}
  \centering
  \includegraphics[width=35mm] 
    {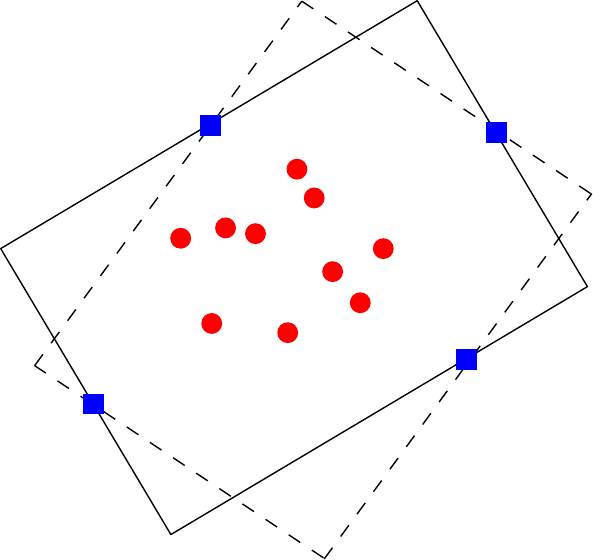}\\
    {\small (a)}\\
    \end{minipage}%
  \begin{minipage}[c]{0.33\textwidth}
  \centering
  \includegraphics[width=35mm]
    {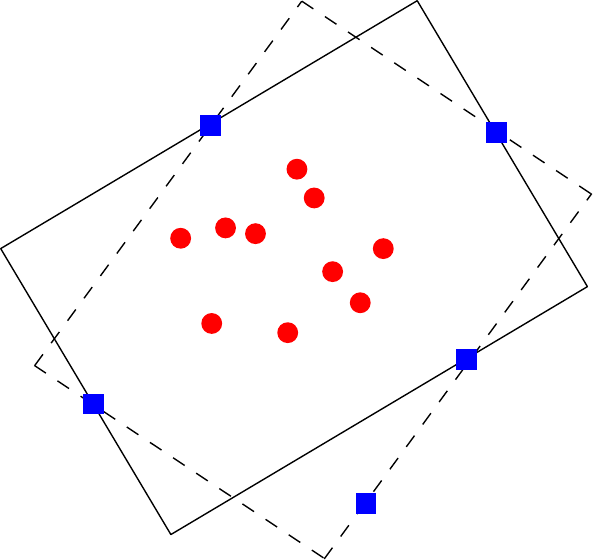}\\
    {\small (b)}\\
    \end{minipage}%
  \begin{minipage}[c]{0.33\textwidth}
  \centering
  \includegraphics[width=35mm]
    {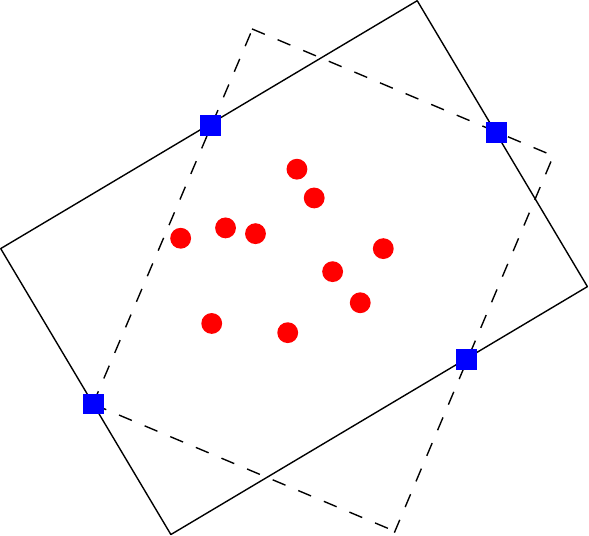}\\
    {\small (c)}\\
    \end{minipage}%
  \caption{The rotated rectangle is shown by dashed sides in all cases. (a) rotation of rectangle by keeping the points fixed, (b) the rectangle touches a point of $B$ which was not on its boundary before rotation, (c) a blue point on a side of the rectangle becomes its corner point during rotation.}
  \label{fig:rectrotate}
\end{figure*}

We compute all the rectangles of $C$ and return one that contains the maximum number of red points. 
Given two points $p,q$ such that $p\in B$ and $q\in R\cup B$, we design a subroutine to compute the rectangles of $C$ such that each of them contains $p$ and $q$ on a single side. 

\subsubsection{The Subroutine}\label{subrtn}
We are given two points $p,q$ such that $p\in B$ and $q\in R\cup B$. We would like to compute the set of rectangles $C_{pq}$ anchored by $p$ and $q$. Without loss of generality, suppose the line through $p$ and $q$ is on the $x$-axis. 
We show how to compute the subset $C'$ of rectangles of $C_{pq}$ whose interior are lying above the $x$ axis. By symmetry, using a similar approach we can compute the rectangles of $C_{pq}$ whose interior are lying below the $x$ axis. We consider all the blue points strictly above the $x$ axis, and we denote this set by $B_1$. 
\begin{figure*}
  \centering
  \includegraphics[width=90mm]
    {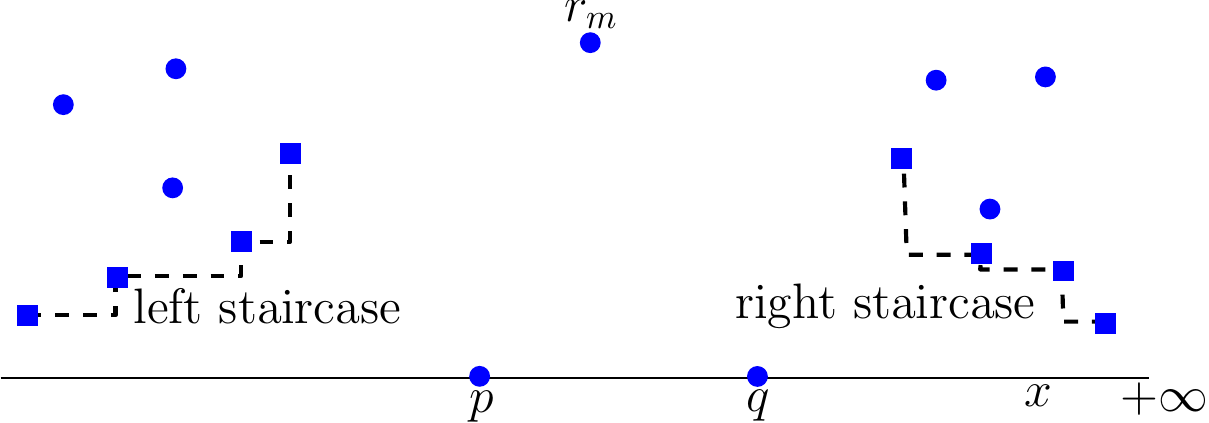}\\
  \caption{Staircase points are shown by squares.}
  \label{fig:staircase}
\end{figure*}

Now consider the points of $B_1$ whose $x$ coordinates are between $x_p$ and $x_q$ (if any). Let $r_m$ be a point having minimum $y$ coordinate among those points. 
%
Let $B'$ be the subset of points of $B_1$ having $y$ coordinate less than $y_{r_m}$. 
%
Define the \textit{left staircase} $B_l$ to be the subset of $B'$, such that for any point $r_2 \in B'$, $r_2$ belongs to $B_l$ if $x_{r_2} < x_p$ and $\nexists r_1 \in B'$ such that $x_{r_2} < x_{r_1} < x_p$ and $y_{r_2} \geq y_{r_1}$ 
%
%
(see Figure \ref{fig:staircase}). Similarly, we define the \textit{right staircase} $B_r$ as follows. For any $r_2$ of $B'$, $r_2 \in B_r$ if $x_q < x_{r_2}$ and $\nexists r_1 \in B'$ such that $x_q < x_{r_1} < x_{r_2}$ and $y_{r_1} \leq y_{r_2}$ (see Figure \ref{fig:staircase}). Also let $B_a= B_l\cup B_r\cup \{r_m\}$. Now we have the following observation.

\begin{obs}
 For any rectangle $T\in C'$, $T$ contains $p,q$ in one of its sides and each of the other sides of $T$ either contains a point of $B_a$ or is unbounded.
\end{obs}


Orlowski \cite{Orlowski90} designed a linear time algorithm for finding the set of rectangles that has one side on $x$-axis and each of the other sides either contains a point of a staircase or is unbounded. 
Our subroutine at first computes the rectangle anchored by $p,q$ and $r_m$ (if any). Then it uses Orlowski's algorithm to compute the remaining rectangles of $C'$.
It also computes the number of red points inside each rectangle it scans by making a query and returns the rectangle that contains the maximum number of red points. The following theorem is due to Goswami~\emph{et al.} \cite{GoswamiDN04}. 

\begin{theorem}\label{th:goswami}
 \cite{GoswamiDN04} For any set of $n$ points, using $O(n^2)$ preprocessing time and space, one can create a data structure that can handle rectangular (of any arbitrary orientation) query in $O(\log n)$ time.
\end{theorem}

We note, that given two sorted lists of blue points in non-decreasing order of $x$ and $y$ coordinates, respectively, by Theorem \ref{th:goswami} the subroutine runs in $O(|C'|\log n)$ time.  
%
Hence we have the following lemma.

\begin{lemma}\label{lem:sbrtn}
 Using $O(n^2+m\log m)$ preprocessing time and $O(n^2+m)$ space a maximum red rectangle of $C_{pq}$ can be computed in $O(|C_{pq}|\log n)$ time.
\end{lemma}

Note that this subroutine can be trivially used for finding a maximum red rectangle in $C$ by calling it for each $p\in B$ and $q\in R\cup B$, and by returning a maximum red rectangle among all such choices. 
But, as we need to compute the sorted lists in every rotated plane (defined by $p$ and $q$) we need $O(m^2(n+m)\log m)$ time in total just for sorting. 
Thus this trivial algorithm runs in $O(m^2(n+m)\log m+n^2 + |C|\log n)$ time. 
In the next subsection, we improve this time complexity by using a careful observation that rescues us from the burden of sorting in every rotated plane. 
\subsubsection{The Improved Algorithm}
The rectangles of $C$ are oriented at angles with respect to the $x$-axis in the range $[0,360)$. Moreover, there are $O(m(n+m))$ orientations (or angles) of these rectangles each defined by a point of $B$ and a point of $R\cup B$. Note that we need the sorted lists of blue points in each such orientation. We use a novel technique for maintaining the sorted lists using some crucial observations. We note that this problem itself might be of independent interest.


%
%

Consider two points $a,b \in B$ and an angle $\theta$. We want to find the ordering of these two points with respect to $x$ coordinates, in the plane oriented counterclockwise at the angle $\theta$. Let us denote this plane by $P_{\theta}$. Also let $B_x^{\theta}$ (resp. $B_y^{\theta}$) be the set of blue points in $P_{\theta}$ sorted in increasing order of $x$ (resp. $y$) coordinates. The following lemma explains how the relative ordering of two points gets changed with changes in plane orientation.

\begin{figure}[h]
\centering
\includegraphics[width=0.5\textwidth]{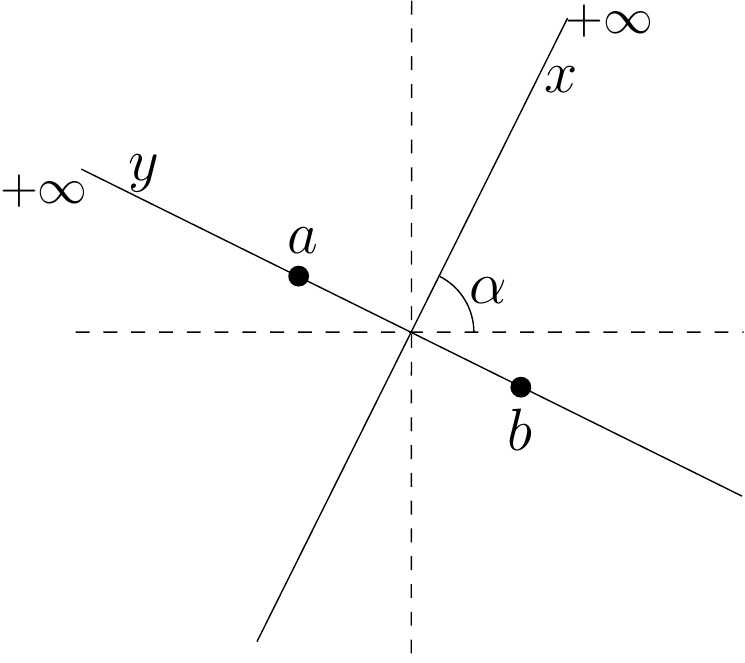}
\caption{The plane generated by rotation of axes by an angle $\alpha$ where $x_a=x_b$.}
\label{fig:rotation}
\end{figure}

\begin{restatable}{lemma}{xrotate}\label{lem:xrotation}
 For any two points $a,b \in B$, there exists an angle $\phi < 180$ such that $x_a = x_b$ in $P_{\phi}$ and $P_{\phi+180}$, and exactly one of the following is true,
 \begin{enumerate}[(i)]
  \item $x_a < x_b$ in $P_{\theta}$ for $\theta \in [0,\phi)\cup (\phi+180,360)$ and $x_a > x_b$ in $P_{\theta}$ for $\theta \in (\phi,\phi+180)$
  \item $x_a > x_b$ in $P_{\theta}$ for $\theta \in [0,\phi)\cup (\phi+180,360)$ and $x_a < x_b$ in $P_{\theta}$ for $\theta \in (\phi,\phi+180)$
 \end{enumerate}
\end{restatable}

\begin{proof}
 Denote the line segment connecting $a$ and $b$ (w.r.t $P_0$) by $l$. We translate $a,b$ in a way so that the midpoint of $l$ is now at the origin. Note that translation does not change the ordering of the points. Let $f_a$ (resp. $f_b$) be the function such that $f_a(\theta)$ (resp. $f_b(\theta)$) is the $x$-coordinate value of $a$ (resp. $b$) in $P_{\theta}$. Now consider a continuous rotation process of the axes (or the plane) in counterclockwise direction with respect to the origin, keeping the points $a,b$ fixed. If $f_a(0)$ is equal to $f_b(0)$, then $a$ and $b$ are on the $y$-axis in $P_0$. Let $\alpha=0$ in this case. Otherwise, $a$ and $b$ are on the opposite sides of the $y$-axis. In this case, as we rotate the axes the value of $|f_a(\theta)-f_b(\theta)|$ becomes zero at some angle $\theta=\alpha$ when both $a$ and $b$ lie on the $y$-axis (see Figure \ref{fig:rotation}). Note that if $f_a(0) < f_b(0)$, $f_a(\theta) < f_b(\theta)$ for all $\theta\in [0,\alpha)$, as $a$ and $b$ do not change their sides with respect to the $y$-axis during this rotation process. Similarly, if $f_a(0) > f_b(0)$, $f_a(\theta) > f_b(\theta)$ for all $\theta\in [0,\alpha)$. Now in both of the cases ($f_a(0)=f_b(0)$ or $f_a(0)\neq f_b(0)$), after a slight rotation $a$ and $b$ will be on the opposite sides of the $y$-axis. As we rotate further the value of $|f_a(\theta)-f_b(\theta)|$ again becomes zero at some angle $\theta=\beta$ when again both $a$ and $b$ lie on the $y$-axis. In the first case, either $f_a(\theta) > f_b(\theta)$ for all $\theta\in (\alpha,\beta)$, or $f_a(\theta) < f_b(\theta)$ for all $\theta\in (\alpha,\beta)$. In the second case, if $a$ was on the left (resp. right) of the $y$-axis in $P_{\theta}$ for $\theta\in [0,\alpha)$, now $a$ is on the right (resp. left) side of the $y$-axis. Thus if $f_a(\theta) < f_b(\theta)$ for  $\theta\in [0,\alpha)$, $f_a(\theta) > f_b(\theta)$ for all $\theta\in (\alpha,\beta)$. Similarly, if $f_a(\theta) > f_b(\theta)$ for  $\theta\in [0,\alpha)$, $f_a(\theta) < f_b(\theta)$ for all $\theta\in (\alpha,\beta)$. As we rotate past $\beta$ in both of the cases, $a$ and $b$ will again be on the opposite sides of the $y$-axis. If $f_a(\theta) < f_b(\theta)$ for $\theta\in (\alpha,\beta)$, $f_a(\theta) > f_b(\theta)$ for all $\theta\in (\beta,360)$. Similarly, if $f_a(\theta) > f_b(\theta)$ for $\theta\in (\alpha,\beta)$, $f_a(\theta) < f_b(\theta)$ for all $\theta\in (\beta,360)$. 
 
 We let $\phi = \alpha$. Note that $f_a(\theta)$ can be equal to $f_b(\theta)$ for $\theta \in [0,360)$ only when $y$-axis is aligned with the line through $a$ and $b$. Also note that during a complete 360$^{\circ}$ rotation of the axes this can happen exactly twice. Moreover, the difference between the corresponding angles should be 180, i.e $\beta = \phi+180$. This completes the proof of the lemma. 
\end{proof}

Similarly, we get the following lemma for ordering of two points with respect to $y$-coordinates. 

\begin{lemma}\label{lem:yrotation}
 For any two points $a,b \in B$, there exists an angle $\phi < 180$ such that $y_a = y_b$ in $P_{\phi}$ and $P_{\phi+180}$, and exactly one of the following is true,
 \begin{enumerate}[(i)]
  \item $y_a < y_b$ in $P_{\theta}$ for $\theta \in [0,\phi)\cup (\phi+180,360)$ and $y_a > y_b$ in $P_{\theta}$, for $\theta \in (\phi,\phi+180)$
  \item $y_a > y_b$ in $P_{\theta}$ for $\theta \in [0,\phi)\cup (\phi+180,360)$ and $y_a < y_b$ in $P_{\theta}$, for $\theta \in (\phi,\phi+180)$
 \end{enumerate}
\end{lemma}


We refer to the angles $\phi$ and $\phi+180$ in Lemma \ref{lem:xrotation} and \ref{lem:yrotation} as \textit{critical angles} with respect to $a$ and $b$. Let $E_1$ (resp. $E_2$) be the set of all critical angles corresponding to pairs of blue points with respect to $x$ (resp. $y$) coordinates. An element of $E_1$ (resp. $E_2$) will be referred to as an event point of first (resp. second) type. Let $A$ be the set of angles (or orientations) defined by pairs of points $(p,q)$ such that $p \in B$ and $q\in R\cup B$. We refer to an element of $A$ as an event point of third type. We construct the set $E=E_1\cup E_2\cup A$ of event points. Also we add the angle $0$ to $E$ as an event point of both first and second type.

Lemma \ref{lem:xrotation} and \ref{lem:yrotation} hint the outline of an algorithm. We sort the list $E$ of event points in increasing order of their angles. We process the event points in this order. 
To start with we construct the set $B_x^0$ (resp. $B_y^0$) at the very first event point which is $0$. In case the algorithm encounters an event point of the first (resp. second) type it swaps the corresponding two points in the current sorted list $B_x^{\phi}$ (resp. $B_y^{\phi}$), where $\phi\in [0,360)$. At an event point of the third type it makes a call to the subroutine of Section \ref{subrtn}. Lastly, a rectangle that contains the maximum red points over all orientations is returned as the solution.

The correctness of this algorithm depends on the correctness of maintaining the sorted lists in all the orientations. As the sorted order with respect to the $x$ (resp. $y$) coordinates do not change in between two consecutive event points of the first (resp. second) type, it is sufficient to prove the following lemma.

\begin{restatable}{lemma}{sorted}\label{lem:sorted}
 At each event point $\phi$ of the first (resp. second) type, the correct sorted list $B_x^{\phi}$ (resp. $B_y^{\phi}$) is maintained. 
\end{restatable}

\begin{proof}
 We prove this lemma for the event points of the first type. The proof for the event points of the second type is similar. We use induction argument on the event points. In the base case for $\phi = 0$ we correctly compute the set $B_x^{0}$. Now assume that for $\alpha \in E_1$ and any $\phi \leq \alpha$ such that $\phi \in E_1$, the correct sorted list $B_x^{\phi}$ is maintained. Let $\beta$ be the successor of $\alpha$ in $E_1$. 
 Let $a$ and $b$ be the two points corresponding to $\beta$. By Lemma \ref{lem:xrotation}, the ordering of $x_a$ and $x_b$ should be different in $P_{\psi}$ and $P_{\theta}$ for $\psi \in (\alpha,\beta)$ and $\theta \in (\beta,\beta+180)$. Thus our algorithm rightly swaps $a$ and $b$ in $B_x^{\alpha}$ to get $B_x^{\beta}$. The location of any other point remains same. Now if there is any point $c$ in $B_x^{\alpha}$ in between $a$ and $b$, then the order of $a$ and $c$ (resp. $b$ and $c$) also gets changed and we end up computing a wrong list. Thus it is sufficient to prove the following claim.

\begin{figure}[h]
\centering
\includegraphics[width=0.6\textwidth]{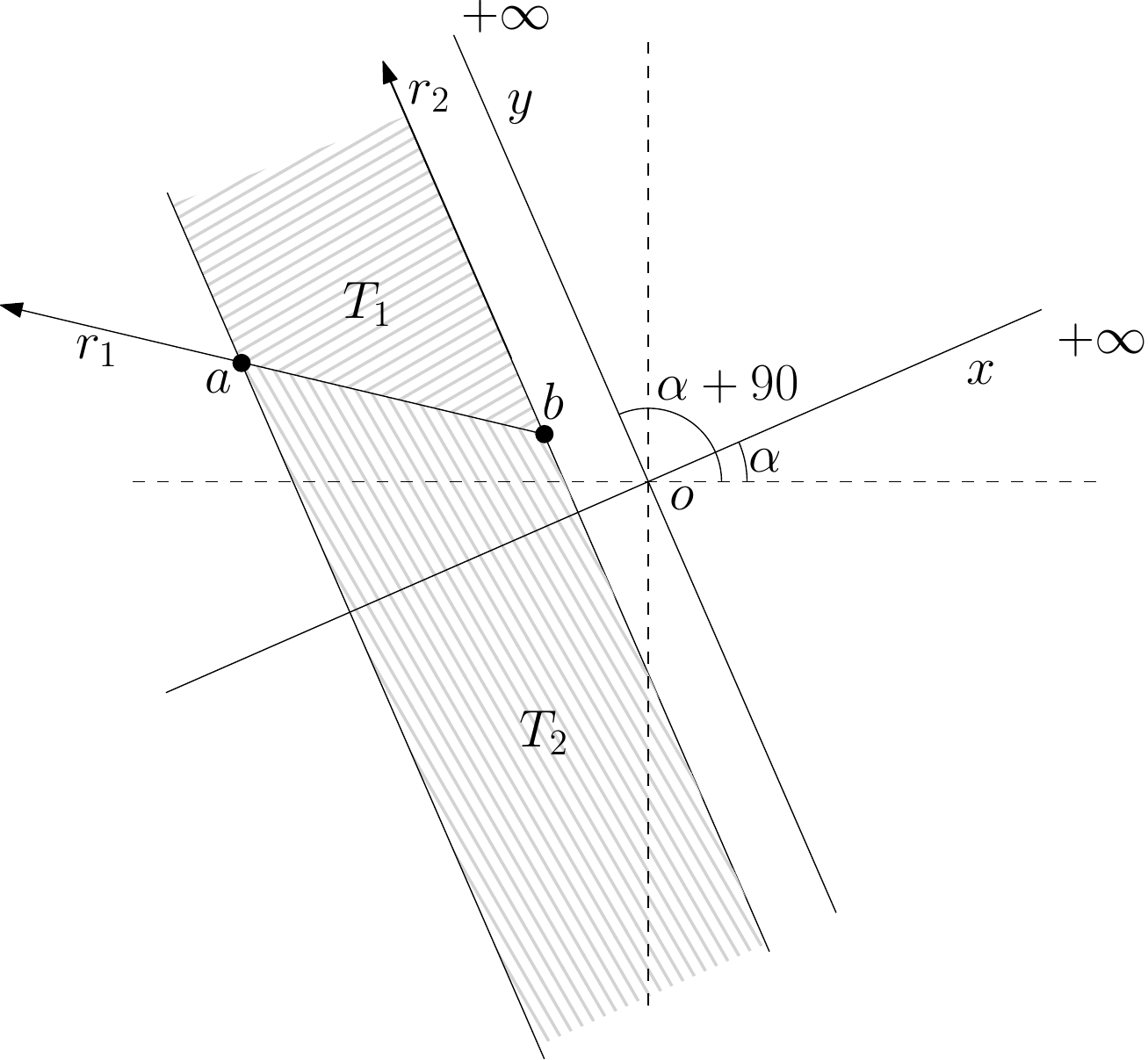}
\caption{The sets $T_1$ and $T_2$. $T_1$ is contained in the wedge defined by $r_1$ and $r_2$.}
\label{fig:claim1}
\end{figure}

\begin{Claim}
$a$ and $b$ must be consecutive in $B_x^{\alpha}$.
\end{Claim}
\begin{proof}
We consider the case where $\beta < 180$. The other case is symmetric. Now consider the plane $P_{\alpha}$. Without loss of generality assume that $x_a \leq x_b$ in this plane. Suppose there is a point $c$ in $B_x^{\alpha}$ in between $a$ and $b$. If $x_a=x_b$, then $x_a=x_c$. But this violates the general position assumption. Thus we consider the case where $x_a < x_b$. Let $l_a$ (resp. $l_b$) be the line passing through $a$ (resp. $b$) which is parallel to the $y$-axis. Let $T$ be the strip between the two lines $l_a$ and $l_b$, i.e $T$ contains all the points $p$ such that $x_a \leq x_p\leq x_b$. Then $c$ must be contained in $T$. Note that the line segment that connects $a$ and $b$ divides $T$ into two sets $T_1$ and $T_2$, above and below $\overline{ab}$ respectively (see Figure \ref{fig:claim1}). Suppose $c$ is in $T_1$. Consider two rays $r_1$ and $r_2$ both originated at $b$ that pass through the point $a$ and $(x_b,y_b+1)$, respectively. Note that $r_1$ and $r_2$ are oriented at angles $\beta+90^{\circ}$ and  $\alpha+90^{\circ}$, respectively. Also the wedge defined by $r_1$ and $r_2$ with angle less than 180$^{\circ}$ contains $T_1$ (see Figure \ref{fig:claim1}). Thus as $c$ is in $T_1$ there is a critical angle $\theta$ with respect to $b$ and $c$ such that $\alpha \leq \theta < \beta$. If $\alpha < \theta < \beta$, we get a contradiction to the assumption that $\beta$ is the next event point in $E_1$. If $\theta$ is equal to $\alpha$, $x_c=x_b$. As $c\in T_1$, $y_c > y_b$. Thus in any plane $P_{\theta}$ with $\theta \in (\alpha,\beta)$, $x_b < x_c$. Hence if $c$ appears before $b$ in $B_x^{\alpha}$, we get a contradiction to the assumption that $B_x^{\alpha}$ is the correct sorted list. In case $c$ is in $T_2$, one can get similar contradiction. Hence the claim follows.
\end{proof}
\end{proof}

Now consider the time complexity of our algorithm. Sorting of $O(m^2+mn)$ event points takes $O((m^2+mn)\log (m+n))$ time. Handling of any event point of first (resp. second) type takes constant time. By Lemma \ref{lem:sbrtn}, handling of all the event points of third type takes in total $O(|C|\log n)$ time. As $|C|=\Omega(m^2+mn)$ we get the following theorem. 

\begin{theorem}
 MPP can be solved in $O(|C|\log (m+n)+n^2)$ time with $O(n^2+m^2)$ space required.
\end{theorem}
\section{Maximum Coloring Problem}\label{sec:mcp}
Recall, that we are given the collection $Q=\{(a_1,b_1),\ldots, (a_n,b_n)\}$ of pairs of points and $P=\cup_{1\leq i\leq n} \{a_i,b_i\}$. A coloring of the points in $P$ is valid if for each pair of points in $Q$, exactly one point is colored by red and the other is colored by blue. For any such valid coloring $C$ and any half plane $h$, we denote the number of red points and the number of blue points in $h$ by $h_r(C)$ and $h_b(C)$ (see Figure \ref{fig:twocol}). For any valid coloring $C$, $\eta(C)=\max |h_r(C)|$, where the maximum is taken over all halfplanes $h$ for which $|h_b(C)|=0$.

\begin{figure*}
  \centering
  \includegraphics[width=50mm]
    {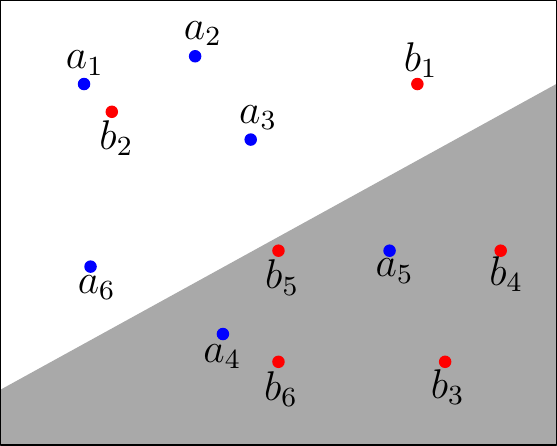}\\
  \caption{An example of a valid coloring. For the halfplane $h$ shown in grey, $h_r(C)=4$ and $h_b(C)=2$.}
  \label{fig:twocol}
\end{figure*}

In the Maximum Coloring Problem (MaxCol), the objective is to find a coloring $C^*$ which maximizes $\eta()$ over all valid colorings. Next we show that a maximum valid coloring can be found by solving the following problem.

\begin{myprob}
 Given a collection $Q=\{(a_1,b_1),\ldots, (a_n,b_n)\}$ of pair of points find a halfplane that contains at most one point from each pair and contains maximum number of points from $P$ where $P=\cup_{1\leq i\leq n}\{a_i,b_i\}$.\label{prob_monotone}
\end{myprob}

The following lemma shows that it is sufficient to solve Problem \ref{prob_monotone} to get a solution for MaxCol.
\begin{restatable}{lemma}{reduce}\label{lem:reduce}
 Suppose there exist an algorithm $\mathcal{A}$ which finds an optimal solution to Problem \ref{prob_monotone} in time $T(n)$. Then MaxCol can be solved in $T(n)+O(n)$ time.
\end{restatable}
\begin{proof}
Let the Algorithm $\mathcal{A}$ outputs the halfplane $h^*$ and $|h^*\cap P|=k$. We prove the claim that given $h^*$, we can find a valid coloring $C^*$, where $\eta(C^*)=k$, which maximizes $\eta()$ over all valid colorings.

We define $C^*$ as follows. Now for each pair $(a_i,b_i)$ such that $|h^*\cap\{a_i,b_i\}|=1$, we color the point in $h^*\cap\{a_i,b_i\}$ by red and the point in $h^*\setminus \{a_i,b_i\}$ by blue. For each pair $(a_i,b_i)$ such that $|h^*\cap\{a_i,b_i\}|=0$, we arbitrarily color one point by red and the other point by blue.
Given $h^*$ such a coloring can be found in $O(n)$ time.

Clearly this coloring is a valid coloring. Next we show that this is maximum as well. Suppose $C^*$  is not maximum and there exists a valid coloring $C$ such that $\eta(C^*) < \eta(C)$. Let $h_C$ be the halfplane corresponding to $C$ that contains only red points and $|P\cap h_C|=\eta(C)$. But then $h_C$ is a half plane such that $h_C$ contains at most one point from each pair in $Q$ and $|P\cap h_C| > |P\cap h^*|$ which contradicts the optimality of $h^*$.  
\end{proof}

%

Instead of directly solving Problem \ref{prob_monotone}, we consider its dual. We use the standard point/line duality that preserves the above/below relationship between points and lines. Let $\mathbb{L}=\{(l_1,\overline{l_1}),\ldots,(l_n,\overline{l_n})\}$ be the collection of the corresponding $n$ pairs of lines in the dual space and let $L$ be the set of those $2n$ lines. Also let $A$ be the set of points $p$ such that for each pair in $\mathbb{L}$, at most one line of the pair lie below $p$. For any $p\in A$, let $\nu(p)$ denote the number of lines in $L$ that lie below $p$. Similarly, let $B$ be the set of points $p$ such that for each pair in $\mathbb{L}$, at most one line of the pair lie above $p$. For any $p\in B$, let $\nu(p)$ denote the number of lines in $L$ that lie above $p$. Using duality, we get the following problem. 
\begin{myprob}
  Given a collection $\mathbb{L}$ of $n$ pairs of lines in the plane, find a point in $A\cup B$ that  maximizes the function $\nu()$.\label{dualprob}
\end{myprob}

Here we describe how to find a point in $A$ that maximizes $\nu()$. The point in $B$ which maximizes $\nu()$ can be found similarly. More specifically we solve the following decision problem.

\begin{myprob}
  Given a collection $\mathbb{L}$ of $n$ pairs of lines in the plane and an integer $k$, does there exist a point $p$ in $A$ such that $\nu(p)=k$.\label{dualprobdecision}
\end{myprob}

If one can solve this decision version in time $T'(n)$, then Problem \ref{dualprob} can be solved in $O(T'(n)\log n)$ time by using a simple binary search on the values of $k$. Observe, that for an ``YES'' instance of Problem \ref{dualprobdecision}, the point $p$ must be on the $k$-level in the arrangement of those $2n$ lines (the $k$-level of an arrangement of a set of lines is the polygonal chain formed by the edges that have exactly $k$ other lines strictly below them). Thus it is sufficient to compute the $k$-level and decide if there is a point $p$ on it such that $p\in A$. Our approach is the following.

Let $\Gamma_k$ be the $k$-level. We traverse the vertices of $\Gamma_k$ in sorted order of their $x$-coordinates. Throughout the traversal we maintain a list $Arr[]$ of size $n$, where $Arr[j]$ denotes the number of lines from $\{l_j,\overline{l_j}\}$ that are currently below the $k$-level. Note that the value of $Arr[j]$ can be $0$, $1$ or $2$. We also maintain an integer $n_b$ which denotes the number of pairs of lines currently below $\Gamma_k$. In other words $n_b$ is the number of $2$'s in $Arr[]$. We update $Arr[]$ and $n_b$ at each vertex of $\Gamma_k$. If a line belonged to a pair $\{l_j,\overline{l_j}\}$ leaves $\Gamma_k$, we reduce the value of $Arr[j]$ by one. Moreover, if the value of $Arr[j]$ was $2$ before, we reduce $n_b$ by one. Similarly, if a line belonged to a pair $\{l_j,\overline{l_j}\}$ enters $\Gamma_k$, we increase the value of $Arr[j]$ by one. If $Arr[j]$ becomes $2$, we also increase $n_b$ by one. At any point during the traversal if $n_b$ becomes zero, we report yes. Otherwise we report no at the end. 

Chan \cite{Chan95} designed an algorithm that computes the $k$-level in an arrangement of $n$ lines in the plane in $O ( n \log b + b^{1+\epsilon})$ time. From the result of Dey \cite{Dey97} we know $b=O(n(k + 1)^\frac{1}{3})$. Thus $k$-level can be computed in $O(n \log n + nk^\frac{1}{3})$ time. Hence we have the following lemma.

\begin{lemma}
 Problem \ref{dualprobdecision} can be solved in $O ( n \log n + n^{1+\epsilon}k^{\frac{1+\epsilon}{3}})$ time.
\end{lemma}

Hence by using a binary search on the values of $k$ we have the following result.

\begin{theorem}
 Problem \ref{prob_monotone} and Problem \ref{dualprob} can be solved in $O (n^{\frac{4}{3}+\epsilon}\log n)$ time.
\end{theorem}
\paragraph{Acknowledgements.} We would like to thank an anonymous reviewer of
an earlier version of this paper for suggestions that has helped us improve the running time of the algorithm for MaxCol.

\bibliographystyle{plain}
\bibliography{bichrome}

\begin{thebibliography}{10}

\bibitem{AggarwalS87}
Alok Aggarwal and Subhash Suri.
\newblock Fast algorithms for computing the largest empty rectangle.
\newblock In {\em SoCG, Waterloo, Canada}, pages 278--290, 1987.

\bibitem{arkincccg}
Esther~M. Arkin, Aritra Banik, Paz Carmi, Gui Citovsky, Matthew~J. Katz, Joseph
  S.~B. Mitchell, and Marina Simakov.
\newblock Conflict-free covering.
\newblock In {\em CCCG, Kingston, Ontario, Canada, August 10-12, 2015}, 2015.

\bibitem{arkin2015choice}
Esther~M Arkin, Aritra Banik, Paz Carmi, Gui Citovsky, Matthew~J Katz,
  Joseph~SB Mitchell, and Marina Simakov.
\newblock Choice is hard.
\newblock In {\em Algorithms and Computation}, pages 318--328. Springer, 2015.

\bibitem{arkin15}
Esther~M. Arkin, Jos{\'{e}}~M. D{\'{\i}}az{-}B{\'{a}}{\~{n}}ez, Ferran Hurtado,
  Piyush Kumar, Joseph S.~B. Mitchell, Bel{\'{e}}n Palop, Pablo
  P{\'{e}}rez{-}Lantero, Maria Saumell, and Rodrigo~I. Silveira.
\newblock Bichromatic 2-center of pairs of points.
\newblock {\em Comput. Geom.}, 48(2):94--107, 2015.

\bibitem{ArmaseluD16}
Bogdan Armaselu and Ovidiu Daescu.
\newblock Maximum area rectangle separating red and blue points.
\newblock In {\em {CCCG} 2016, British Columbia, Canada, August 3-5, 2016},
  pages 244--251, 2016.

\bibitem{AronovH08}
Boris Aronov and Sariel Har{-}Peled.
\newblock On approximating the depth and related problems.
\newblock {\em {SIAM} J. Comput.}, 38(3):899--921, 2008.

\bibitem{BackerK10}
Jonathan Backer and J.~Mark Keil.
\newblock The mono- and bichromatic empty rectangle and square problems in all
  dimensions.
\newblock In {\em {LATIN} 2010, Oaxaca, Mexico, April 19-23, 2010.}, pages
  14--25, 2010.

\bibitem{BackerK09}
Jonathon Backer and J.~Mark Keil.
\newblock The bichromatic square and rectangle problems.
\newblock In {\em Technical Report 2009-01, University of Saskatchewan}, 2009.

\bibitem{Bar-NoyCS08}
Amotz Bar{-}Noy, Panagiotis Cheilaris, and Shakhar Smorodinsky.
\newblock Deterministic conflict-free coloring for intervals: From offline to
  online.
\newblock {\em {ACM} Transactions on Algorithms}, 4(4), 2008.

\bibitem{BiniazBMS16}
Ahmad Biniaz, Prosenjit Bose, Anil Maheshwari, and Michiel H.~M. Smid.
\newblock Plane bichromatic trees of low degree.
\newblock In {\em Combinatorial Algorithms - 27th International Workshop,
  {IWOCA} 2016, Helsinki, Finland, August 17-19, 2016, Proceedings}, pages
  68--80, 2016.

\bibitem{BiniazMNS15}
Ahmad Biniaz, Anil Maheshwari, Subhas~C. Nandy, and Michiel H.~M. Smid.
\newblock An optimal algorithm for plane matchings in multipartite geometric
  graphs.
\newblock In {\em WADS 2015, Victoria, BC, Canada, August 5-7, 2015.}, pages
  66--78, 2015.

\bibitem{BitnerCD10}
Steven Bitner, Yam~Ki Cheung, and Ovidiu Daescu.
\newblock Minimum separating circle for bichromatic points in the plane.
\newblock In {\em {ISVD} 2010, Quebec, Canada, June 28-30, 2010}, pages 50--55,
  2010.

\bibitem{Chan95}
Timothy~M. Chan.
\newblock Output-sensitive results on convex hulls, extreme points, and related
  problems.
\newblock In {\em SOCG, Vancouver, B.C., Canada, June 5-12, 1995}, pages
  10--19, 1995.

\bibitem{ChaudhuriND03}
Jeet Chaudhuri, Subhas~C. Nandy, and Sandip Das.
\newblock Largest empty rectangle among a point set.
\newblock {\em J. Algorithms}, 46(1):54--78, 2003.

\bibitem{ChazelleDL86}
Bernard Chazelle, Robert L. (Scot)~Drysdale III, and D.~T. Lee.
\newblock Computing the largest empty rectangle.
\newblock {\em {SIAM} J. Comput.}, 15(1):300--315, 1986.

\bibitem{CheilarisGRS14}
Panagiotis Cheilaris, Luisa Gargano, Adele~A. Rescigno, and Shakhar
  Smorodinsky.
\newblock Strong conflict-free coloring for intervals.
\newblock {\em Algorithmica}, 70(4):732--749, 2014.

\bibitem{ChenFKLMMPSSWW07}
Ke~Chen, Amos Fiat, Haim Kaplan, Meital Levy, Jir{\'{\i}} Matousek, Elchanan
  Mossel, J{\'{a}}nos Pach, Micha Sharir, Shakhar Smorodinsky, Uli Wagner, and
  Emo Welzl.
\newblock Online conflict-free coloring for intervals.
\newblock {\em {SIAM} J. Comput.}, 36(5):1342--1359, 2007.

\bibitem{CortesDPSUV09}
C.~Cort{\'{e}}s, Jos{\'{e}}~Miguel D{\'{\i}}az{-}B{\'{a}}{\~{n}}ez, Pablo
  P{\'{e}}rez{-}Lantero, Carlos Seara, Jorge Urrutia, and Inmaculada Ventura.
\newblock Bichromatic separability with two boxes: {A} general approach.
\newblock {\em J. Algorithms}, 64(2-3):79--88, 2009.

\bibitem{Cristianini}
Nello Cristianini and John Shawe-Taylor.
\newblock {\em An Introduction to Support Vector Machines: And Other
  Kernel-based Learning Methods}.
\newblock Cambridge University Press, New York, NY, USA, 2000.

\bibitem{Dey97}
Tamal~K. Dey.
\newblock Improved bounds on planar k-sets and k-levels.
\newblock In {\em FOCS, Miami Beach, Florida, USA, October 19-22, 1997}, pages
  156--161, 1997.

\bibitem{Duda}
Richard~O. Duda, Peter~E. Hart, and David~G. Stork.
\newblock {\em Pattern Classification (2Nd Edition)}.
\newblock Wiley-Interscience, 2000.

\bibitem{EcksteinHLNS02}
Jonathan Eckstein, Peter~L. Hammer, Ying Liu, Mikhail Nediak, and Bruno
  Simeone.
\newblock The maximum box problem and its application to data analysis.
\newblock {\em Comp. Opt. and Appl.}, 23(3):285--298, 2002.

\bibitem{EvenLRS03}
Guy Even, Zvi Lotker, Dana Ron, and Shakhar Smorodinsky.
\newblock Conflict-free colorings of simple geometric regions with applications
  to frequency assignment in cellular networks.
\newblock {\em {SIAM} J. Comput.}, 33(1):94--136, 2003.

\bibitem{GoswamiDN04}
Partha~P. Goswami, Sandip Das, and Subhas~C. Nandy.
\newblock Triangular range counting query in 2d and its application in finding
  k nearest neighbors of a line segment.
\newblock {\em Comput. Geom.}, 29(3):163--175, 2004.

\bibitem{kaneko}
Atsushi Kaneko and M.~Kano.
\newblock Discrete geometry on red and blue points in the plane — a survey
  —.
\newblock In {\em Discrete and Computational Geometry}, volume~25 of {\em
  Algorithms and Combinatorics}, pages 551--570. Springer Berlin Heidelberg,
  2003.

\bibitem{katz2012conflict}
Matthew~J Katz, Nissan Lev-Tov, and Gila Morgenstern.
\newblock Conflict-free coloring of points on a line with respect to a set of
  intervals.
\newblock {\em Computational Geometry}, 45(9):508--514, 2012.

\bibitem{LiuN03}
Ying Liu and Mikhail Nediak.
\newblock Planar case of the maximum box and related problems.
\newblock In {\em CCCG'03, Halifax, Canada, August 11-13, 2003}, pages 14--18,
  2003.

\bibitem{Naamad84}
A~Naamad, D.T Lee, and W.-L Hsu.
\newblock On the maximum empty rectangle problem.
\newblock {\em Discrete Applied Mathematics}, 8(3):267 -- 277, 1984.

\bibitem{Orlowski90}
M.~Orlowski.
\newblock A new algorithm for the largest empty rectangle problem.
\newblock {\em Algorithmica}, 5(1):65--73, 1990.

\end{thebibliography}

\end{document}